\documentclass[letterpaper,twocolumn,10pt,journal]{IEEEtran}

\usepackage{amssymb}
\usepackage{graphicx}
\usepackage{amsmath}
\usepackage{amsthm}
\usepackage[linesnumbered,algoruled,vlined]{algorithm2e}
\usepackage{algcompatible}
\usepackage{balance}

\usepackage{comment}

\usepackage{caption}
\usepackage{subcaption}
\usepackage{xcolor}

\newtheorem{theorem}{Theorem}
\newtheorem{lemma}[theorem]{Lemma}

\usepackage{xparse} 
\usepackage{array,multirow}

\makeatletter
\newcommand{\removelatexerror}{\let\@latex@error\@gobble}
\makeatother

\newcounter{ipCounter}
\NewDocumentEnvironment{IPFormulation}{m}{%
\refstepcounter{ipCounter}
\begin{algorithm}[#1]%

}{%
\end{algorithm}
\addtocounter{algocf}{-1}
}

\NewDocumentEnvironment{IPFormulationStar}{m}{%
\refstepcounter{ipCounter}
\begin{algorithm*}[#1]%

}{%
\end{algorithm*}
\addtocounter{algocf}{-1}
}

 \newcolumntype{F}{>{$\displaystyle\,}r<{$}@{\hspace{0.0em}}}
 \newcolumntype{C}{>{$\displaystyle\,}c<{$}@{\hspace{0.0em}}}
 \newcolumntype{B}{>{$\displaystyle\,}r<{$}@{\hspace{0.0em}}}
 \newcolumntype{R}{>{$\displaystyle}r<{$}@{\hspace{0.2em}}}
 \newcolumntype{S}{>{$\displaystyle}r<{$}@{\hspace{0.2em}}}
 \newcolumntype{L}{>{$\displaystyle}l<{$}@{\hspace{0.2em}}}
 \newcolumntype{Q}{>{$\displaystyle}l<{$}@{\hspace{0.3em}}}

 \newcommand{\tagIt}[1]{\refstepcounter{equation}\textnormal{({\theequation})} \label{#1}}

\begin{document}

\title{Approximate and Incremental\\Network Function Placement}

\author{Tam\'as Lukovszki \quad Matthias Rost \quad Stefan Schmid\\
{\small E\"otv\"os Lor\'and University, Hungary \quad 
TU Berlin, Germany \quad
Aalborg University, Denmark
}
}
 

\newcommand{\capa}{\kappa}



\maketitle

\sloppy

\begin{abstract}
The virtualization and softwarization
of modern computer networks
introduces
interesting new opportunities for
a more flexible placement of 
network functions and middleboxes (firewalls, proxies, traffic optimizers, 
virtual switches, etc.).
This paper studies approximation algorithms for
the incremental deployment of a minimum number of middleboxes at optimal locations,
such that capacity constraints at the middleboxes and length constraints on the communication routes are respected.
Our main contribution is a new, purely combinatorial and 
rigorous proof for the submodularity of the function maximizing the number of communication requests that can be served by a given set of middleboxes. 
 Our proof allows us to devise a deterministic approximation algorithm which uses
an augmenting path approach to compute the submodular function. 
This algorithm does not require any changes to the locations of existing
middleboxes or the preemption of previously served communication pairs when additional middleboxes are deployed, previously accepted communication pairs just can be handed over to another middlebox. It is hence particularly attractive for incremental deployments.
We prove that the achieved polynomial-time approximation bound is optimal, 
unless $P=\mathit{NP}$. 
This paper also initiates the study of a weighted problem variant,
in which entire groups of nodes need to communicate via a 
middlebox (e.g., a multiplexer or a shared object), possibly at different rates.
We present an LP relaxation and randomized rounding algorithm for this problem,
leveraging an interesting connection to scheduling. 
\end{abstract}

\section{Introduction}\label{sec:intro}

Middleboxes are ubiquitous in modern computer networks
and provide a wide spectrum of in-network functions related to 
security, performance, and policy compliance.
It has recently been reported that 
the number of middleboxes in enterprise networks can be of the same order of
magnitude as the number of routers~\cite{someone}.

While in-network functions were traditionally implemented in specialized hardware
appliances and middleboxes,
computer networks in general and middleboxes in particular become more and more
software-defined and virtualized~\cite{opennf}:
network functions can be implemented in software and deployed fast and flexibly
on the virtualized network nodes, e.g., running in a virtual machine
on a commodity x86 server.
Modern computer networks also offer new flexibilities in terms of how
traffic can be \emph{routed} through middleboxes and virtualized data plane appliances (often
called Virtual Network Functions, short \emph{VNFs})~\cite{MeasuRouting}.
In particular, the advent of \emph{Software-Defined Network (SDN)}
technology allows network operators
to steer traffic through middleboxes (or chains of middleboxes) using arbitrary routes, i.e., along routes
which are not necessarily shortest paths, or
not even loop-free~\cite{reviewer-1,reviewer-2,flowtags,simple}.
In fact, OpenFlow, the standard SDN protocol today,
not only introduces a more flexible routing, but
itself allows to implement basic middlebox functionality, on the switches~\cite{road}:
an OpenFlow switch can match, and perform actions upon, not only
layer-2, but also layer-3 and layer-4 header fields.

However, not much is known today about how to exploit these
flexibilities algorithmically.
A particularly interesting problem regards the question of
where to deploy a minimum number of middleboxes such that basic routing and capacity
constraints are fulfilled. Intuitively,
the smaller the number of deployed network functions,
the longer the routes via these functions,
and a good tradeoff between
deployment costs and additional latency must be found.
Moreover, ideally, middleboxes should be \emph{incrementally
deployable}: when additional middleboxes are deployed,
existing placements do not have to be changed. 
This is desirable especially in deployment scenarios with budget constraints, where an existing deployment has to be extended by new middleboxes.


\subsection{Our Contributions}

We initiate the study of the natural
problem of (incrementally) placing
a minimum number of middleboxes or network functions.

Our main technical result is a deterministic and greedy (polynomial-time)
$O(\log{(\min\{\kappa,n\})})$-approximation algorithm for the (incremental)
middlebox placement problem in $n$-node networks where
capacities are bounded by $\kappa$. 
The algorithm is attractive for incremental deployments: 
it does not require any changes to the locations of existing
middleboxes or the preemption of previously served communication pairs when additional middleboxes are deployed.

At the heart of our algorithm lies a new and purely combinatorial proof of the
submodularity of the function maximizing the number of pairs that can be served by a given set of middleboxes. 
The submodularity proof directly implies a 
deterministic approximation algorithm for the minimum
middlebox deployment problem.
We show that the derived approximation bound
is asympotically optimal in the class of all polynomial-time algorithms, 
unless
$P=\textit{NP}$.

This paper also initiates the discussion of two generalizations
of the network function placement problem: (1) a 
 model where not only node pairs but entire node \emph{groups}
need to be routed via certain network functions (e.g., a multiplexer or 
a shared object in a distributed cloud),
and (2) a weighted model where nodes have arbitrary resource requirements.
By leveraging a connection to scheduling,
we show that these problems can also be approximated efficiently. 

We believe that our model and approach has ramifications
beyond middlebox deployment. For instance, 
our model also captures fundamental problems arising in the context of incremental SDN deployment (solving an open problem in~\cite{panopticon-conf,icnp15shear})
or distributed cloud computing where resources need to be allocated for large user groups.

We also initiate the study of a weighted problem variant,
where entire groups of nodes need to communicate via a 
shared network function (e.g., a multiplexer of a multi-media
conference application or a shared object), possibly at different rates.
We present an LP relaxation and randomized rounding algorithm for this problem,
establishing an interesting connection to scheduling. 

\subsection{Novelty and Related Work}

In this paper we are interested in algorithms
which provide formal approximation guarantees. 
In contrast to classic covering problems~\cite{Chvatal79,Feige98,wolsey-submodular}:
(1)~we are interested in the distance \emph{between communicating pairs}, \emph{via} the covering
nodes, and not \emph{to} the covering nodes;
(2)~we aim to support \emph{incremental deployments}: 
middlebox locations selected earlier in time
as well as the supported communication pairs
should not have to be changed when deploying additional middleboxes;
(3)~
we consider
a \emph{capacitated} setting where the number of items which can be assigned to a node is bounded by~$\kappa$.

To the best of our knowledge, 
so far, only heuristics or informal studies without
complete or combinatorial proofs of the approximability~\cite{midas,panopticon,mehraghdam2014specifying,someone}
as well as algorithms with an exponential runtime in the worst-case~\cite{HSVB15,panopticon},
 of the 
(incremental and non-incremental) middlebox deployment problem
have been presented in the literature.
Our work also differs from function chain embedding 
problems~\cite{sirocco16path,sirocco15,sss16moti} which often revolve around other
objectives such as maximum request admission or minimum link load.

Nevertheless, we in this paper show that we can build upon
hardness
results on uncapacitated covering problems~\cite{LundYannanakis94}
as well as Wolsey's study of vertex and set covering problems
with hard capacities~\cite{wolsey-submodular}.
An elegant alternative proof to
Wolsey's dual fitting approach, based on
combinatorial arguments, is due to
Chuzhoy and Naor~\cite{covering-capa}.
In~\cite{covering-capa} the authors also show that using LP-relaxation approaches is generally difficult,
as the integrality gap of a natural linear program for
the weighted and capacitated vertex and set covering problems
is unbounded.

\noindent \textbf{Bibliographic Note.} A preliminary version of
this paper appeared in CCR~\cite{ccr16match}.

\subsection{Putting Things into Perspective}\label{sec:idea}

Our problem is a natural one and, as mentioned, 
features some interesting connections to and has implications for
classic
(capacitated) problems. 
To make things more clear and put things into perspective, 
in the following, we will elaborate more on this relationship.

The generalized dominating set problem~\cite{Bansal11} asks for a set of
dominating nodes of minimal cardinality, such that the
distance from any network node to a dominator is at most $\ell$. In the capacitated version
of the problem, the number
of network nodes which can be dominated by a node is limited.
Similarly, capacitated facility location problems~\cite{Aardal15,Chudak05}
ask for locations to deploy facilities, subject to
capacity constraints,
such that the number of facilities as well as the distance to the facilities
are jointly optimized.
In contrast to these problems where the distance \emph{to} a dominator or facility
\emph{from a given node} is measured, in our middlebox deployment problem,
we are interested in the distance (resp.~stretch) 
\emph{between node pairs}, \emph{via} the middlebox.
At first sight, one may intuitively expect that optimal solutions to
dominating set or facility location problems are also good approximations
for the middlebox deployment problem. However, this is not the case, as we illustrate
in the following. Consider an optimal solution to the distance $\ell/2$ dominating set problem:
since the distance to any dominator is at most $\ell/2$, the locations of the dominators
are also feasible locations for the middleboxes: the route between two nodes via the
dominator is at most $\ell/2+\ell/2=\ell$.
However, the number of dominators in this solution can be much larger than the number
of required middleboxes.

Consider the example in
Figure~\ref{fig:worst-case}: in a star network where communicating node pairs are located
at the leaves, at depth 1 and $\ell-1$, and where node capacities are sufficiently high such that
they do not constitute a bottleneck,
a \emph{single} middlebox $m$ in the center is sufficient. However, the
stricter requirement that dominating nodes must be at distance at most $\ell/2$,
results in a dominating set of cardinality
$\Omega(n/\ell)$, i.e., $\Omega(n)$ for constant $\ell$.

\begin{figure}[t]
\centering
\includegraphics[width=0.8\columnwidth]{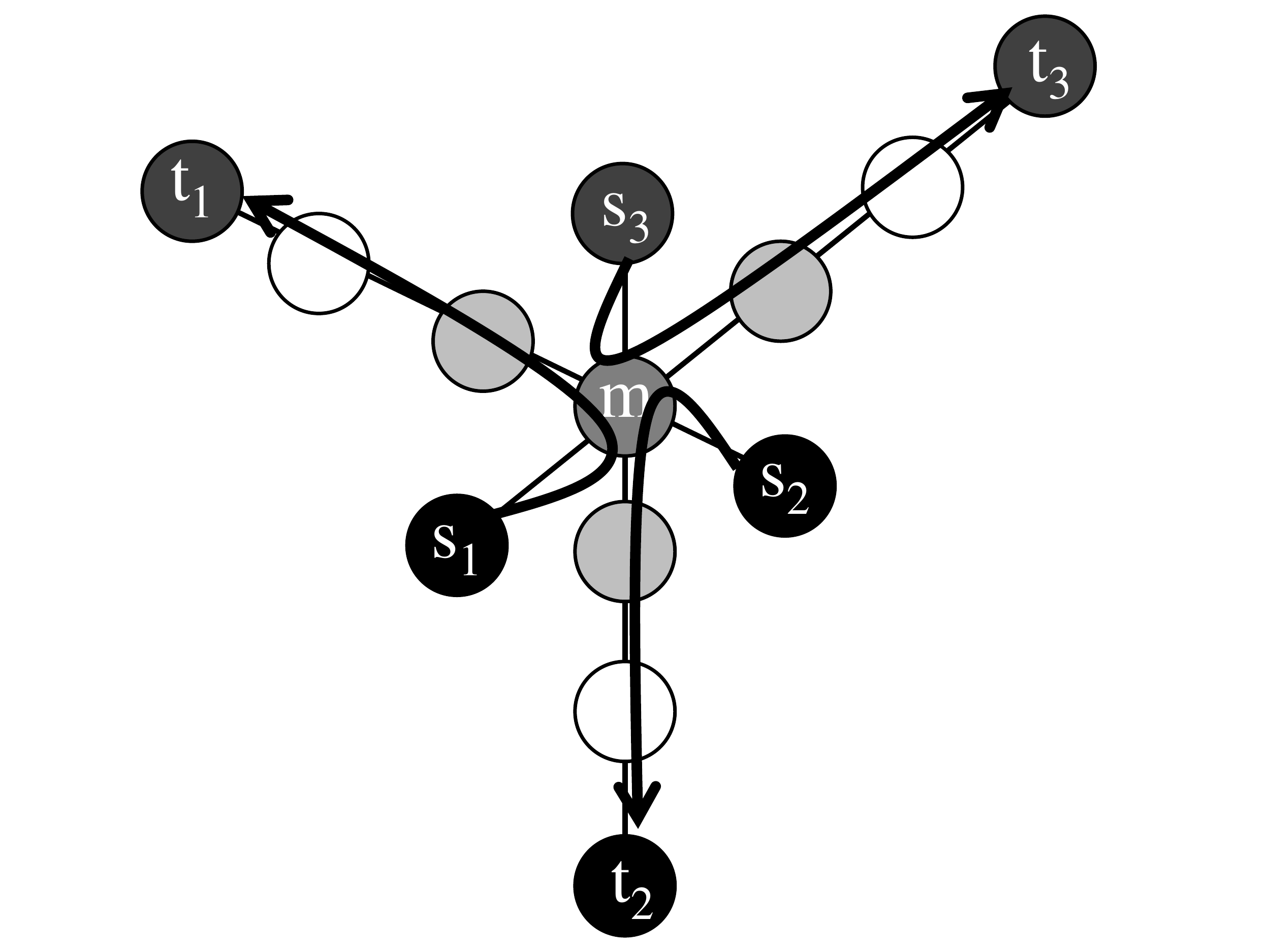}
\caption{Example: A minimum cardinality dominating set (three nodes, in \emph{light grey}) is a bad approximation for the
middlebox deployment problem (one node in center, in \emph{darker grey}).}
\label{fig:worst-case}
\end{figure}


While our discussion revolved around constraints 
on the length $\ell$, similar arguments and bounds also
apply to the \emph{stretch}, the main focus of this paper:
the ratio of the length of the path through a middlebox, 
and the length of the shortest path between the communicating pair.
To see this, simply modify the example in
Figure~\ref{fig:worst-case} by adding a direct edge between 
each communicating pair $s_i$ and $t_i$. Let $c=\ell$. Then the length of 
the shortest path between each pair is one. 
The dominating nodes must be at distance at most $c/2$ from the communicating nodes. Consequently, the cardinality of the dominating set is 
$\Omega(n/c)$, i.e., $\Omega(n)$ for constant $c$;
deploying one single middlebox $m$ at the center and routing the communication 
between each pair through $m$ results in paths of stretch~$c$.

\begin{figure}[t]
\centering
\includegraphics[width=0.8\columnwidth]{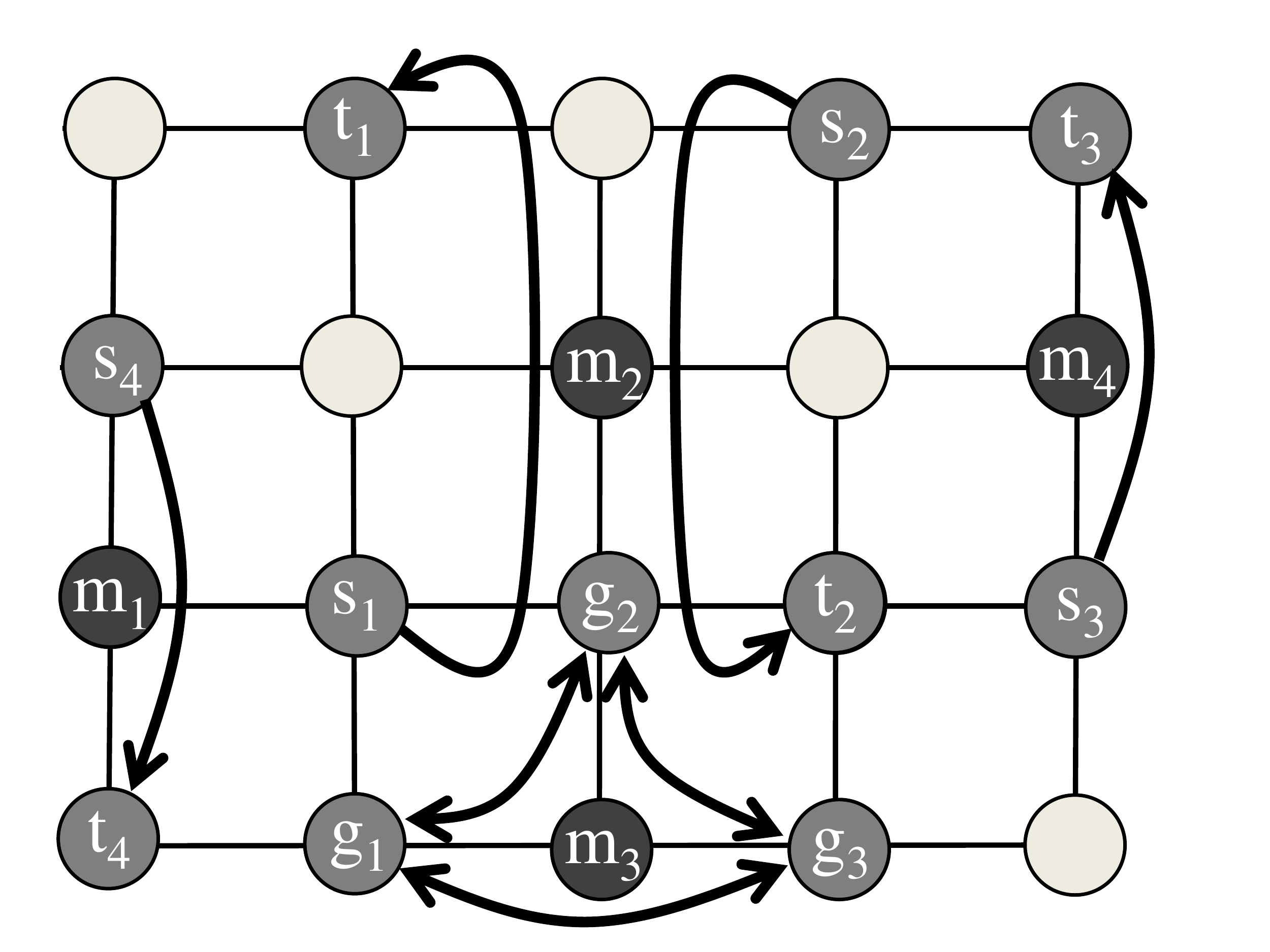}
\caption{Example: The communication between node pairs $(s_i,t_i)$, for $i\in \{1,2,3,4\}$,
as well as between the group of nodes $\{g_1,g_2,g_3\}$
needs to be routed via a network function resp.~middlebox $M$.
Due to capacity constraints and constraints on the route length,
$M$ is instantiated at four locations $\{m_1,m_2,m_3,m_4\}$ in this example.}\label{fig:example}
\end{figure}

\subsection{Organization}

The remainder of this paper is organized as follows.
Section~\ref{sec:model} introduces our model and
discusses different use cases.
In Section~\ref{sec:approximation} we present
our approximation algorithm together with its analysis.
 Section~\ref{sec:groups} extends our study
 to weighted and group models. 
After reporting on our simulation results in
 Section~\ref{sec:eval}, we conclude our
 contribution in  Section~\ref{sec:summary}.

\section{Model and Use Cases}\label{sec:model}

\subsection{Formal Model}

We model the computer network as a graph connecting a set $V$ of $n=|V|$
nodes.
The input to our problem is a set of communicating node
pairs $P$: the route of each
node pair $(s,t)\in P$ needs to traverse an instance
of a middlebox resp.~network function (e.g., a firewall).
Node pairs
do not have to be disjoint: nodes can participate in many
communications simultaneously.

For the sake of generality, we assume that middleboxes can only be installed
on a subset of nodes $U\subseteq V$.
We will refer to the set of middleboxes locations (and equivalently,
the set of middleboxes instances) by $M$, and
we are interested in deploying a minimal number of middleboxes
at legal locations $M\subseteq U$
such that:
\begin{enumerate}
\item Each pair $p=(s,t)\in P$ is assigned
to an instance $m \in M$, denoted by $m=\mu(p)$.
\item
For each pair $p=(s,t)\in P$, there is a route from $s$ via $m=\mu(p)$
to $t$ of length at most $\rho \cdot d(s,t)$, i.e., $d(s,m)+d(m,t)\leq \rho \cdot d(s,t)$,
where $d(u,v)$ denotes the length of the shortest path between
nodes $u,v\in V$ in the network $G$, and where $\rho\geq 1$ is called the \emph{stretch}. For example, a stretch $\rho = 1.2$ implies that the distance (and the latency accordingly) when routing via a middlebox is at most 20\% more than the distance (and the latency accordingly) when connecting directly.
Our approach supports many alternative constraints, e.g., on the maximal route length. 
\item Capacities are respected: at most $\capa$ node pairs can
be served by any middlebox instance.
\end{enumerate}

Our objective is to
minimize the number of required middlebox instances, subject to the above constraints.

In this paper, we will also initiate the study of a weighted model,
where different communication pairs have different demands,
as well as a group model, where network functions need to serve entire groups of 
communication partners. 
See Figure~\ref{fig:example} for an example.
 
\subsection{Use Cases}

Let us give three concrete examples motivating
our formal model.

\noindent\textbf{Middlebox Deployment.}
Our model is mainly
motivated by the middlebox placement 
flexibilities introduced in network function virtualized and software-defined
networks. 
Deploying additional middleboxes, network processors
or so-called ``universal nodes'' can be costly, and a good tradeoff should
be found between deployment cost and routing efficiency.
For example, today, network policies can often be defined in 
terms of adjacency matrices or big switch abstractions, 
specifying which traffic is allowed between an ingress
port $s$ and an outgress network port $t$. In order to enforce
such a policy, traffic from $s$ to $t$ needs to traverse a middlebox instance 
inspecting and classifying the 
flows. The location of every middlebox can be optimized, but is subject to
the constraint that the route from $s$ to $t$ via the middlebox
should not be much longer than the shortest path from $s$ to~$t$.

\noindent \textbf{Deploying Hybrid Software-Defined Networks.}
There is a wide consensus that the transition of existing
networks to SDN will not be instantaneous,
and that SDN technology will be deployed incrementally, for cost reasons
and to gain confidence. 
The incremental OpenFlow switch deployment
problem can be solved using waypointing (routing flows via OpenFlow switches); our paper solves the algorithmic problem in the Panopticon~\cite{panopticon}
system.

\noindent\textbf{Distributed Cloud Computing.}
The first two use cases discussed above come with
per-pair requirements: each communicating pair
must traverse at least one function. However, there are also
scenarios where entire groups $G_i$ of nodes need
to share a waypoint:
for example, a multiplexer of a group communication (a multi-media conference) or a shared object in a distributed system, e.g., a shared and collaborative editor: 
An example could be a distributed
cloud application: imagine a set of users $G_i$ who
would like to use a collaborative editor application \`{a} la Google Docs.
The application should be hosted on a server which is
located close to the users, i.e., minimizing the latency
between user pairs.

\section{Approximation Algorithm}\label{sec:approximation}

\subsection{Overview}

This section presents a deterministic and polynomial-time 
$O(\log(\min\{n,\kappa\}))$-approximation algorithm
for the middlebox deployment problem.
Our algorithm
is based on an efficient computation of a certain submodular set function:
it defines the maximum number of pairs which can be covered by a given set of middleboxes.
In a nutshell, the submodular function
is computed efficiently using an augmenting path method on a certain bipartite graph,
which also determines the corresponding assignment of communication pairs to the middleboxes.
The augmenting path algorithm is based on a simple,
linear-time breadth-first
graph traversal. The augmenting path method
is attractive and may be of independent interest: similarly to the flow-based
approaches in the literature~\cite{NaorN02}, it does not require changes to previously
deployed middleboxes, but removes the disadvantage of~\cite{NaorN02} that the set of served
communication pairs changes over time: an attractive property for \emph{incremental
deployments}.

Our solution with augmenting paths results (theoretically and practically) in significantly faster
computations of the submodular function value as well as of the corresponding assignment
compared to the
flow based approach.
In fact, computing all augmenting paths takes $O(|E| \min\{\sqrt{|V|}, \kappa\})$
by using the Hopcroft-Karp algorithm \cite{HopcroftKarp73} performing at most
$O(\min\{\sqrt{|V|}, \kappa\})$ breadth-first search and depth-first search traversals of the graph.
Computing a maximum flow takes $O(|V||E|^2)$ time by the Edmonds-Karp algorithm~\cite{EdmondsKarp72},
$O(|V|^3)$ time by the push-relabel algorithm of Goldberg and Tarjan~\cite{GoldbergTarjan88}, and
$O(|V||E|)$ time by the algorithm of Orlin~\cite{Orlin13} (see~\cite{GoldbergTarjan88} for an overview on efficient maximum flow algorithms).

Concretely, we can start with an empty set of
middlebox locations $M=\emptyset$,
and in each step, we add a middlebox at location $m$ to $M$, which maximizes the number of
totally covered (by middleboxes at locations $M\cup\{m\}$) communicating pairs, without violating capacity
and route stretch constraints.
The algorithm terminates, when all pairs were successfully
assigned to middlebox instances in $M$.
More precisely, for any set of middlebox instances $M\subseteq U$, we define
$\phi(M)$ to be the maximum number of communication pairs that can be assigned to
$M$, given the capacity constraints at the nodes and the route stretch constraints.
We show that $\phi$ is  non-decreasing and submodular, and $\phi(M)$
-- and the corresponding assignment of pairs to middlebox instances in $M$ -- can be computed
in polynomial time.
This allows us to use Wolsey's Theorem~\cite{wolsey-submodular}
to prove an approximation factor of $1+O(\log \phi_{\max})$, where $\phi_{\max} = \max_{m\in U} \phi(\{m\})$.
Since in our case, $\phi_{\max}=\min\{\kappa, |P|\}$ and $|P|\leq n^2$,
this implies that
Algorithm \ref{alg:augment} computes an $O(\log (\min\{\kappa,n\}))$-approximation for the minimum
number of middlebox instances that can cover all pairs $P$
(i.e., all pairs can be assigned to the deployed middleboxes).

\subsection{Maximum Assignment}

In order to compute function $\phi(M)$, for any $M\subseteq U$, we construct a bipartite
graph $B(M)=(M \cup P, E)$,
where $P$ is the set of communicating pairs.
We will simply refer to the middlebox instances $m\in M$ and
pairs $p\in P$ in the bipartite graph as the \emph{nodes}.
The edge set $E$ connects middlebox instances $m\in M$ to those communicating pairs $p\in P$
which can be routed via $m$ without exceeding the stretch constraint,
i.e. $E=\{(m,p)\ :\ m\in M, p=(s,t)\in P, d(s,m)+d(m,t)\leq\rho\cdot d(s,t)\}$, where $d(u,v)$
denotes the length of shortest path between nodes $u$ and $v$ in the network.
For each $p$, the set of such middlebox nodes can be computed in a pre-processing step
by performing an all-pair shortest paths algorithm to calculate $d(u,v)$
for each $u,v$ in the network, and for each $p=(s,t)\in P$,
selecting the nodes $m\in M$ with $d(s,m)+d(m,t)\leq\rho\cdot d(s,t)$.


A \emph{partial assignment} $A(M)\subseteq B(M)$ of pairs $p\in P$ to middlebox instances in $M$
is a subgraph of $B(M)$, in which each $p\in P$ is connected to at most one middlebox
$m\in M$ by an edge, i.e., $\deg_{A(M)}(p)\leq 1$ where $\deg$ denotes the degree.
A pair $p\in P$ with $\deg_{A(M)}(p)=1$ is called an \emph{assigned pair} and with
$\deg_{A(M)}(p)=0$ an \emph{unassigned pair} of a \emph{free pair}.
A partial assignment $A(M)$ without free pairs is called an \emph{assignment}.
The \emph{size} $|A(M)|$ of a (partial) assignment $A(M)$ is defined as
the number of edges in $A(M)$.

Our goal is to compute a partial assignment $A(M)$ of pairs $p\in P$ to middlebox
instances in $M$
maximizing the number of assigned pairs.
Accordingly, we distinguish between \emph{assignment edges} $E_A$ and
\emph{non-assignment edges} $E_{\overline{A}}$,
where $E_A\cup E_{\overline{A}}=E$ is a partition of the edge set $E$ of $B(M)$.

Our algorithm ensures that at any moment of time, the partial assignments are \emph{feasible},
i.e., the assignment fulfills the following capacity constraints.
The \emph{current load} of a middlebox $m$ in $M$, denoted by $\lambda(m)$, is
the number of communicating pairs served by $m$ according to the
current partial assignment $A(M)$.
Moreover, we define the \emph{free capacity} $\capa^*(m)$ of $m$ to be
$\capa^*(m)=\capa-\lambda(m)$.
A (partial) assignment $A(M)$ is feasible if and only if it does not violate 
capacities, i.e., $\lambda(v)\leq \capa$, for all $v$ in any middlebox in $M$.

\begin{figure}[t]
\centering
\includegraphics[width=0.75\columnwidth]{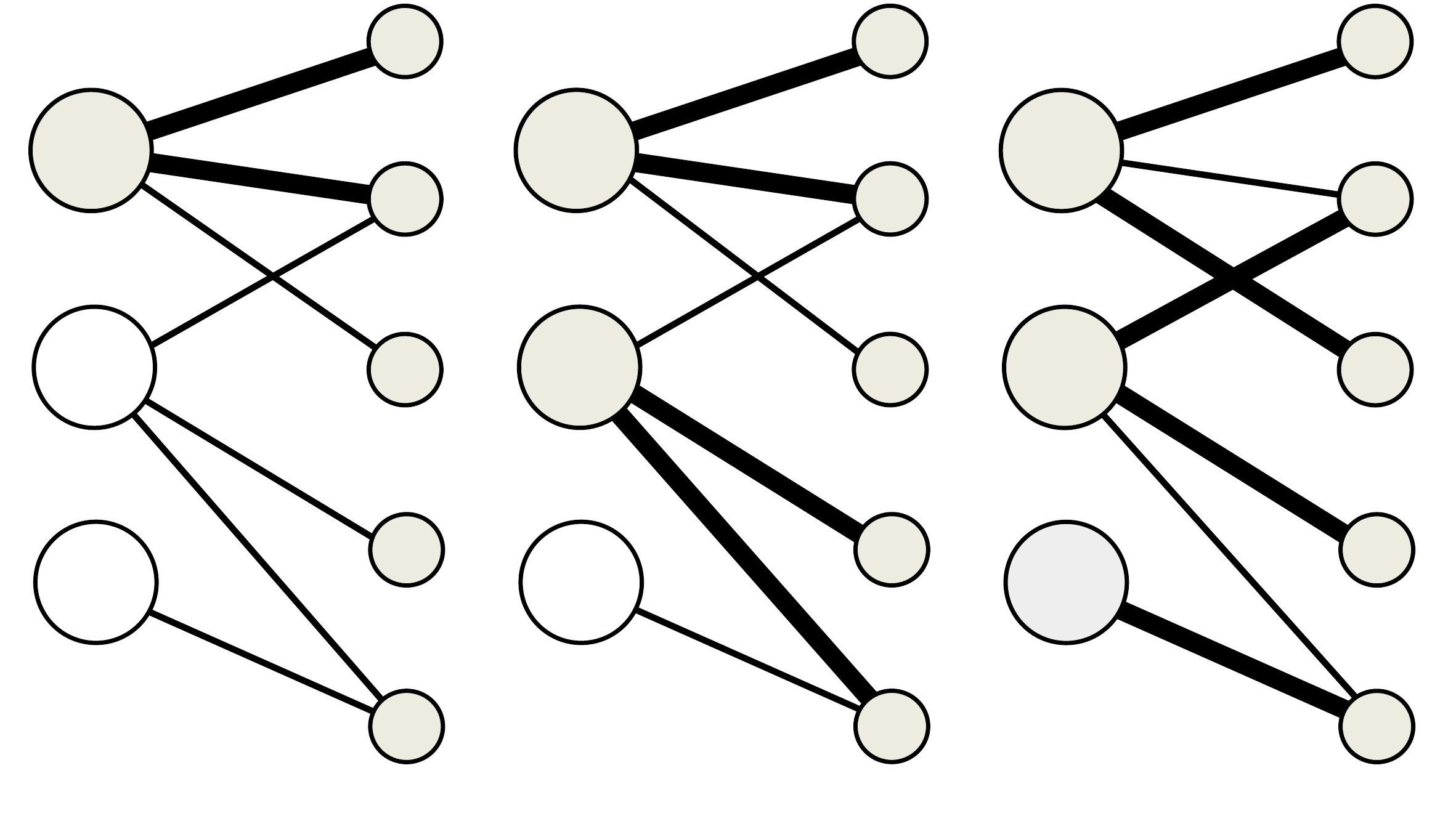}
\caption{
Bipartite graph with possible middlebox locations $U$ (empty: \emph{white}, deployed: \emph{grey}) on the left side and
pairs $P$ on the right side of the bipartite graph. The capacity of the middleboxes $\kappa=2$.
Middleboxes are deployed greedily, one-by-one, without
requiring relocation of previously mapped middleboxes. 
Assignment edges are indicated in \emph{bold}.
The deployment of the first middlebox \emph{left} creates two assigment edges incident to the middlebox. 
The deployment of the second middlebox \emph{middle} involves two new assigment edges incident to the second middlebox. Lastly, the botttom middlebox is used to serve the bottom pair, which was previously served by the middle middlebox.}\label{fig:greedy-steps}
\end{figure}

\begin{figure}[t]
\centering
\includegraphics[width=0.75\columnwidth]{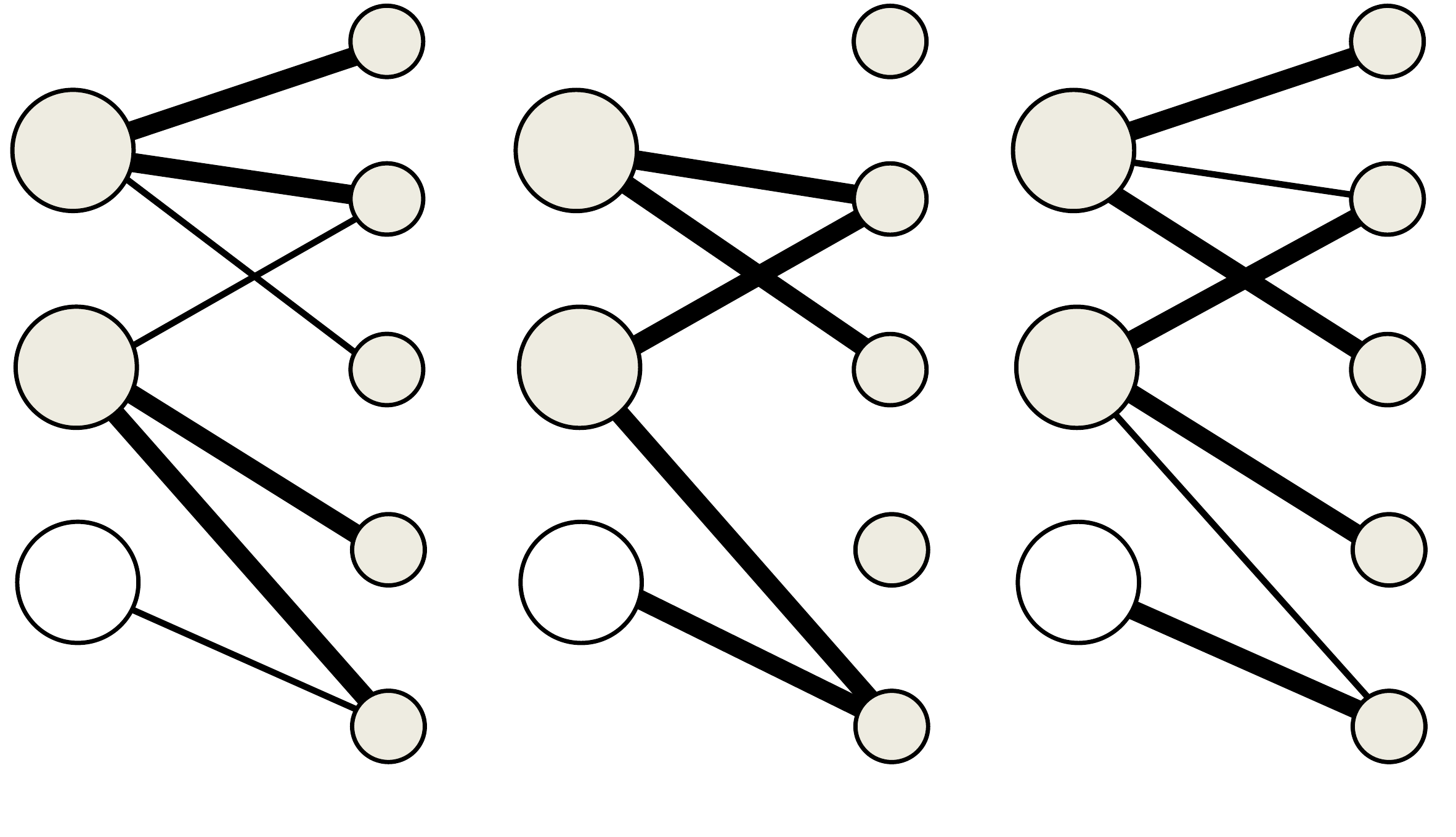}\\
\caption{Illustration of augmenting path computation. Bipartite graph with possible middlebox locations $U$ (empty: \emph{white}, deployed: \emph{grey}) on the left side and pairs $P$ on the right side of the bipartite graph.
The capacity of the middleboxes $\kappa=2$.
Assignment edges are indicated in \emph{bold}.
\emph{Left:} The assignment before deploying the third middlebox. \emph{Middle:} Augmenting path starting at the new middlebox. \emph{Right}: The resulting new assignment.}\label{fig:aug-path}
\end{figure}

In order to compute the integer function $\phi(M)$, which is essentially the cardinality of a maximum feasible partial assignment $A^*(M)$,
we make use of augmenting paths.
 Let $A(M)$ be a feasible partial assignment.
An \emph{augmenting path} $\pi=(v_1,v_2,\ldots,v_{j})$ relative to $A(M)$
in $B(M)$ starts at a middlebox $m\in M$ with free capacity, ends at a free pair $p\in P$, and
alternates between assignment edges and non-assignment edges, i.e.,
\begin{enumerate}
\item $v_1\in M$ with $\capa^*(v_1)>0$ and $v_j\in P$ with $\deg_{A(M)}(v_j)=0$,
\item  $(v_i,v_{i+1}) \in B(M)\setminus A(M)$, for any odd $i$, 
\item  $(v_i,v_{i+1})\in A(M)$, for any even $i$. 
\end{enumerate}

An augmenting path relative to $A(M)$ is a witness for a better partial assignment:
The symmetric difference $A'(M) = (A(M)\setminus \pi) \cup (\pi \setminus A(M))$
is also a (partial) assignment with size $|A'(M)| = |A(M)| + 1$.
Due to the properties of the augmenting path,
by this reassignment, one additional pair will be covered by the same set of middlebox instances,
without violating node capacities: in $A(M)$, the first node had free capacity and the
last node represents a free pair.
Furthermore, the degree at each internal node of $\pi$ remains unchanged, since
one incident assignment edge gets unassigned and one incident unassigned edge gets assigned.
Therefore, the load of the internal nodes in $\pi$ remains unchanged.
Conversely, suppose that a partial assignment $A(M)$ is not a maximum partial assignment.
Let $A^*(M)$ be a maximum partial assignment.
Consider the bipartite graph $X=(A(M)\cup A^*(M))\setminus (A(M)\cap A^*(M))$.

Figure~\ref{fig:greedy-steps} illustrates the greedy placement
strategy and Figure~\ref{fig:aug-path} shows the augmenting
path computation.

Note that an augmenting path must always exist for suboptimal
 covers. To see this, consider the following reduction to matching:
 replace each middlebox node with $\kappa$ many clones of capacity 1,
 and assign each pair $p\in P$ to the clones in the canonical way.
 Feasible assignments constitute a matching in this graph,
 with node degrees one. Given any suboptimal cover, we find a witness for a augmenting
  path as follows. We take the symmetric difference of the suboptimal and the optimal
  solution, which gives us a set of paths and cycles of even length. Thus,
  a path must exist where the optimal solution has one additional edge.

Augmenting paths can be computed efficiently, by simply performing
breadth-first searches in $B(M)$, by using the nodes $m\in M$ with free capacities as the starting nodes.

\subsection{Submodularity}

The set function $\phi: 2^{U} \rightarrow \mathbb{N}$ is called \emph{non-decreasing}
iff $\phi(U_1) \leq \phi(U_2)$ for all $U_1 \subseteq U_2 \in 2^U$,
and \emph{submodular} iff $\phi(U_1)+\phi(U_2) \geq \phi(U_1 \cap U_2) + \phi(U_1 \cup U_2)$
for all $U_1, U_2 \in 2^U$. Equivalently, submodularity can be defined
as follows (See, e.g.~in \cite{Schrijver03} pp.~766): 
for any $U_1,U_2 \subseteq U$ with $U_1\subseteq U_2$ and every $u\in U\setminus U_2$,
we have that $\phi(U_1\cup \{u\})-\phi(U_1)\geq \phi(U_2\cup\{u\})-\phi(U_2)$.
This is in turn equivalent to: for every $U_1\subseteq U$ and $u_1,u_2 \in U \setminus U_1$
we have that $\phi(U_1\cup\{u_1\})+\phi(U_1\cup\{u_2\})\geq \phi(U_1\cup\{u_1,u_2\})+\phi(U_1)$.

Let $U_1 \subseteq U_2$ be two arbitrary subsets of $U$.
Consider a maximum assignment $A(U_1)$ for $U_1$.
Let $A(U_2)$ be a maximum assignment for $U
_2$ obtained from $A(U_1)$ by
adding the members $u_2 \in U_2 \setminus U
_1$ to $U_1$ one-by-one, and
performing the augmenting path method until we have a maximum assignment
for the incremented set $U_1 \cup \{u_2\}$.
Let $A^{\Pi}(U
_1)$ be the projection of $A(U_2)$ to $U
_1$, i.e., for each $u_1 \in U_1$, $p \in P$,
the pair $p$
is assigned to $u_1$ in $A^{\Pi}(U
_1)$ if and only if $p$ is assigned to $u_1$ in $A(U_2)$.
First we show that $A^{\Pi}(U_1)$ is a maximum assignment for $U_1$.

\begin{lemma}\label{lemma:helper}
Let $A(U_1)$ be a maximum assignment for $U_1$.
Let $A(U_2)$ be a maximum assignment for $U
_2$ obtained from $A(U_1)$ by
adding all $u \in U_2 \setminus U_1$ to $U_1$ one-by-one, and
performing the augmenting path method
until we have a maximum assignment for $U_1 \cup \{u\}$.
Let $A^{\Pi}(U
_1)$ be the projection of $A(U_2)$ to~$U_1$.
Then $A^{\Pi}(U
_1)$ is a maximum assignment for~$U_1$.
\end{lemma}
\begin{proof}
By adding $u_2 \in U_2 \setminus U
_1$ to $U_1$ and
performing the augmenting path method until an augmenting path exists
(i.e., until we obtain a maximum assignment for $U_1 \cup \{u_2\}$), it holds that,
for each $u_1 \in U
_1$, the number of pairs assigned to $u_1$ does not change:
along the augmenting path each internal node has one incident assignment
edge and one non-assignment edge.
This also holds after exchanging the assignment edges and the non-assignment edges, i.e.,
each assignment edge on the augmenting path becomes a non-assignment edge and vice versa.
The degrees of the start and end nodes increases by one.
Consequently, the degree of each $u_1 \in U_1$ in $A^{\Pi}(U_1)$ is the same as in $A(U_1)$.
Since $A(U_1)$ is a maximum assignment and each $u_1$ has the same degree in
$A^{\Pi}(U_1)$, $A^{\Pi}(U_1)$ must be also a maximum assignment.
\end{proof}

\begin{theorem}\label{thm:submod}
Let $U$ be the set of all possible middlebox instance locations.
Let $\phi : 2^{U} \rightarrow \mathbb{N}$ be the set
function, such that for $M
\subseteq U$, $\phi(M)$ is
the maximum number of pairs in $P$ that can be assigned to $M$
without violating the capacity constraints. Then $\phi$ is submodular.
\end{theorem}
\begin{proof}
We show that
for all $M
_1,M_2 \subseteq U$,  $M_1 \subseteq M
_2$ and for all
$m\in U \setminus M_2$ we have that
\begin{eqnarray}
\phi(M
_1 \cup \{m\}) - \phi(M_1) \geq \phi(M_2 \cup \{m
\}) - \phi(M_2).
\end{eqnarray}
This is equivalent to the definition of the submodularity of $\phi$.


Consider a maximum assignment $A(M
_1)$ for $M_1$
and a maximum assignment $A(M_2)$ for $M_2$ obtained from $A(M_1)$ as described in Lemma~\ref{lemma:helper}.
Now we add $m$ to $M_2$.
Let $A(M_2 \cup \{m\})$ be the maximum assignment for $M_2 \cup \{m\}$ obtained from $A(M_2)$
by adding $m$ to $M_2$ and
performing the augmenting path method until we have a maximum assignment for $M_2 \cup \{m\}$.
Let $A^{\Pi}(M_2)$ be the projection of $A(M_2 \cup \{u\})$ to $M_2$ and
let $A^{\Pi}(M_1)$ be the projection of $A(M_2 \cup \{u\})$ to $M_1$.
Let $|A^{\Pi}(M_2)|$ and $|A^{\Pi}(M_1)|$ be the number of assigned pairs of $P$ in the assignments.
By Lemma~\ref{lemma:helper}, $A^{\Pi}(M_2)$ is a maximum assignment for $M_2$ and $A^{\Pi}(M_1)$
is a maximum assignment for $M_1$. Therefore, $\phi(M_2) = |A^{\Pi}(M_2)|$
and $\phi(M_1) = |A^{\Pi}(M_1)|$.
Furthermore, $\phi(M_2 \cup \{u
\}) = |A(M_2 \cup \{u\})|$.

Consider the pairs $P_{m} \subseteq P$ assigned to $m$
in the assignment $A(M_2 \cup \{m\})$. Let $A^{\Pi}(m)$
be the projection of $A(M_2 \cup \{m\})$
to $m$. Then $|A^{\Pi}(m)| = |P_{m
}|$.
Since $A^{\Pi}(M_2)$ contains all elements
that are assigned to any element of $M_2$ in
$A(M_2 \cup \{m\})$, clearly
$|A(M_2 \cup \{m\})| = |A^{\Pi}(M_2)| + |A^{\Pi}(m)|$, and thus
\begin{eqnarray}
\phi(M_2 \cup \{m\}) - \phi(M_2) = |A^{\Pi}(m)|.\label{eqsubmod01}
\end{eqnarray}
On the other side, the assignment $A^*(M_1 \cup \{m\})$ which is obtained as the union of
$A^{\Pi}(M_1)$ and
$A^{\Pi}(m)$ is a valid assignment for $M_1 \cup \{m\}$. Therefore,
$|A(M_1 \cup \{m\})| \geq |A^*(M_1 \cup \{m\})| = |A^{\Pi}(M_1)| + |A^{\Pi}(m)|$,
and thus
\begin{eqnarray}
\phi(M_1 \cup \{m\}) - \phi(M_1) \geq |A^{\Pi}(m)|.\label{eqsubmod02}
\end{eqnarray}
Using (\ref{eqsubmod01}) and (\ref{eqsubmod02}) we obtain
\begin{eqnarray}
\phi(M_1 \cup \{m\}) - \phi(M_1) \geq \phi(M_2 \cup \{m\}) - \phi(M_2)\,,
\end{eqnarray}
completing our proof.
\end{proof}

\subsection{The Algorithm}

Essentially, Algorithm 1 starts with an empty set $M$
and cycles through the possible middlebox locations $m\in U\setminus M$,
always deploying the middlebox resulting (with the already deployed
ones) in the highest function value $\phi$.

Given the submodularity and the augmenting path construction,
we have derived our main result.
Per middlebox, an augmenting paths problem is solved. 
Using the \emph{Hopcroft-Karp algorithm}, we can compute
all (at most $\kappa$ many) augmenting paths starting at a newly added middlebox
in time $O(\min\{\kappa,\sqrt{|V|}\}\cdot |E|)$, where $|V|$ denotes the number of nodes
and $|E|$ the number of edges in $B(M)$.

\begin{algorithm}[t]
    \caption{Greedy Algorithm}
    \label{alg:augment}
    
    \begin{algorithmic}[1]
			\STATE \textbf{init} $M\gets \emptyset$, $A(M)\gets$ empty assignment
            \STATE \textbf{while} $A(M)$ is not a feasible assignment \textbf{do}
    		\STATE ~~~\textbf{init} $m^* \gets \emptyset$, $opt \gets 0$, $tmp \gets 0$
 			\STATE ~~~\textbf{for each} $m\in U\setminus M$
 			\STATEx ~~~~~~\emph{(* compute all augmenting paths *)}
 			\STATE ~~~~~~~$tmp \gets \phi(M
  \cup \{m\}) - \phi(M)$ 
 			\STATE ~~~~~~~\textbf{if} $tmp>opt$ \textbf{then}
 			\STATE ~~~~~~~~~~~$opt\gets tmp$, $m^*\gets m$
	        \STATE ~~~~~~~\textbf{end if}
            \STATE ~~~\textbf{end for}
            \STATE ~~~$M
            \gets M\cup\{m^*\}$, \textbf{update} $A(M)$
			\STATE \textbf{end while}
    \end{algorithmic}
\end{algorithm}

\begin{theorem}\label{thm:log-approx}
Our greedy and incremental middlebox deployment algorithm computes a $O(\log{n})$-approximation.
\end{theorem}

\subsection{Lower Bound and Optimality}

Theorem~\ref{thm:log-approx} is essentially the 
best we can hope for:

\begin{theorem}
The middlebox deployment problem
is NP-hard and cannot
be approximated within $c\log n$, for some $c > 0$
unless $P = NP$. Furthermore, it is not
approximable within $(1-\epsilon)\ln n$,
for any $\epsilon > 0$, unless NP  $\subset\mbox{\sc Dtime}(n^{\log\log n})$.
\label{theorem:hardness}
\end{theorem}

\begin{proof}
We present a polynomial time reduction from the Minimum
Set Cover (MSC) problem, defined as follows:
Given a finite set $S$ of $n$ elements and a collection $C$ of subsets of $S$.
A set cover for $S$ is a subset $C'\subseteq C$ such that
every element in $S$ is contained in at least one member
of $C'$. The objective is to minimize the cardinality of the set cover $C'$.

Consider an instance of the MSC problem:
let $S=\{v_1,...,v_n\}$ be a set of $n$
elements, $C = \{S_i\subseteq S, i=1,...,m\}$.
We define the instance of the corresponding
middlebox deployment problem in a network $G=(V,E)$
with a set of communicating pairs $P$ and stretch $\rho=1$ 
as follows. For each element $v\in S$,
we introduce two nodes $v_s$ and $v_t$ in $V$. 
For each subset $S_i\in C$, we introduce a node $v_{S_i}$ in $V$ as well.
The edge set $E$ of the network $G=(V,E)$ is defined
by the following rule: there is an edge $(v_s,v_{S_i})\in E$
and an edge $(v_{S_i},v_t)\in E$ iff the corresponding
element $v$ is contained in $S_i$. 
The set of communicating pairs is defined as $P=\{(v_s,v_t) : v\in S\}$
and the set of potential middlebox locations is defined as $U=\{v_{S_i}:S_i\in C\}$.
$G=(V,E)$ is a bipartite graph with partitions $U$ 
and $\{v_s:v\in S\}\cup\{v_t:v\in S\}$.
If $v\in S$ is contained in a set $S_i\in C$ then there is a
path of length 2 between the corresponding pair $(v_s,v_t)$ in~$G$.
This is also the shortest path between $v_s$ and $v_t$.
In the middlebox deployment problem with stretch $\rho=1$, a
set of nodes $M\subseteq U$ of minimum cardinality must be selected
such that between each pair $(v_s,v_t)\in P$ there is a route of
length of at most~2 and it contains at least one node of $M$. By
the construction of the network, for each pair $(v_s,v_t)\in P$,
there is a route $v_s,v_{S_i},v_t$ of length 2 in $G$ if and only
if $v\in S_i$. Let  $M\subseteq U$
be a minimum cardinality solution of the middlebox deployment problem.
The node set $M$ implies a minimum cardinality solution for the MSC problem
and vice versa.
This proves the NP-hardness of the problem.

The inapproximability results follow from the combination of the
above reduction and the inapproximability results of the minimum
set cover problem by Raz and Safra~\cite{RazSafra97} and
by Feige~\cite{Feige98}. Raz and Safra~\cite{RazSafra97} proved that
the minimum set cover problem is not approximable within $c\log n$,
for some $c > 0$, unless $P = NP$. Feige~\cite{Feige98} showed the
inapproximability within  $(1-\epsilon )\ln n$, for any $\epsilon >0$,
unless  $NP \subset\mbox{\sc Dtime}(n^{\log\log n})$.
If we had a better approximation for the middlebox deployment problem,
by the above reduction, we would have a better approximation factor
for the minimum set cover problem, as well.
\end{proof}

\section{Group and Weighted Variant}\label{sec:groups}

There are scenarios in which more than two nodes may have to share
a network function. For example, consider the problem of placing a multiplexer
for a group of users involved in a teleconference. Or imagine the problem of mapping
a shared object or entire virtual server
server of a multi-user game: in order to avoid state synchronization
overheads, the users should be served from a single location
which is also located close to all the users. Also in the context of
distributed cloud computing, the problem of placing functionality 
for larger groups may be relevant. 
Moreover, so far we have assumed that all requests induce the same 
load. However, in reality, different communication pairs communicate
at different rates, which may also result in different loads on
the middleboxes. 

In this section, we show that both the group assignment
variant as well as the weighted pair request variant
can be solved by exploiting an interesting 
connection to energy-efficient scheduling~\cite{fleischer-subm,monien-scheduling,khuller}. 
As we will see,  
in order to make this generalization, 
we will however need to sacrifice the incremental property.
Moreover, we need a constant factor resource augmentation (concretely,
a factor of 2).

Let $n$ be the total number of requests (set $R$),
and let $p_j$ be the ``price'', i.e., the load induced by assigning
the $j$-th request to a middlebox: the request
can be an arbitrarily weighted node pair
or group of nodes. 
%
Our greedy algorithm runs in rounds, where in 
each round, the ``most effective'' middlebox, a middlebox at a location
which 
maximizes a certain (submodular) function,
is chosen and added to our solution set.
Concretely, given a set $S$ of already deployed middleboxes,
let $F(S)$ denote
the maximum number of requests that can 
be processed by a set $S$ of middleboxes of capacity $\kappa$.
Since computing $F(S)$ is NP-hard, we follow the work of~\cite{khuller} and compute a fractional solution
using linear programming. The resulting fractional values
will eventually be rounded.
 For every middlebox-request pair $(i,j)$, 
 we introduce a binary variable $x_{i,j}$:
 it is 1, if request $j$ is
assigned to middlebox location $i$, and 0 
otherwise. 
As a preprocessing step, let us first delete (1) all
requests $j$ with $p_j>\kappa$, i.e.,
requests which cannot be served by
any middlebox anyway; (2) all request-middlebox
pairs which would exceed stretch constraints;
and (3) all request-middlebox pairs where the middlebox
is not at a legal location. 
We set the corresponding
$x_{i,j}$ values to 0. 
We will relax the $x_{i,j}$ variables,
and define $f(S)$ to be the maximum weighted sum of requests
that can
be fractionally processed by a set $S$ of 
middleboxes.
We can compute $f$ with Linear Program (LP)~\ref{alg:maximum-fractional-assignment}:

\begin{figure}[tbhp]
 {
 
 \vspace{-4pt}
 
 \def\arraystretch{1.3}
 \removelatexerror
 
  \LinesNotNumbered
  \begin{IPFormulation}{H}
  \SetAlgorithmName{LP}{}{{}}
\noindent
\begin{tabular}{RFRLQ}
\textrm{max} & \multicolumn{4}{L}{ f(S) =  \sum_{i,j} x_{i,j} } \\
\textnormal{subject to} & \sum_{i\in U} x_{i,j} & \leq & 1  &\qquad \forall\ \text{$j\in R$}  \\ 
& \sum_{ j \in R} p_j x_{i,j} & \leq  & \kappa   &\qquad \forall\ i\in S \\
& x_{i,j}& = & 0 & \qquad \forall ~ i \in S, j \in R: p_{i,j} > \kappa \\
& \multicolumn{3}{C}{0 \leq x_{i,j} \leq 1} & \qquad \forall\ i \in S, j \in R  
\end{tabular}
  \caption{Maximum Fractional Assignment}
  \label{alg:maximum-fractional-assignment}
  \end{IPFormulation}
  }
\end{figure}

\begin{algorithm}[t]
    \caption{Generalized Greedy Algorithm}
    \label{alg:gen-augment}
    \begin{algorithmic}[1]
			\STATE \textbf{init} $S\gets \emptyset$
            \STATE \textbf{while} $|S|<|U|$ and $f(S)\leq n-1$ \textbf{do}
    		\STATE ~~~\textbf{choose} $i\in U\setminus S$ s.t.~$gain(i,S)$ is 
    		maximized
 			\STATE ~~~$S = S \cup \{i\}$
\STATE \textbf{select middleboxes in set $S$}
\STATE \textbf{round $f(S)$ to integer solution}
			\STATE \textbf{end while}
    \end{algorithmic}
\end{algorithm}

It follows from~\cite{fleischer-subm}
that the relaxed function $f$ is submodular. We define
$gain(i, S) = f (S \cup \{i\}) - f (S)$ for any $S\subseteq U$, $i\in U$. 

The greedy algorithm (Algorithm~\ref{alg:gen-augment})
adds middleboxes one-by-one, greedily selecting the 
one yielding the highest $gain(i,S)$.
We stop when $f(S)> n-1$: stopping prematurely
is needed due to the fractional solution of the relaxed approach~\cite{khuller}. 
In the end, 
we round the fractional solution using 
the approach by Shmoys and Tardos~\cite{tardos-shmoys}.

Given this transformation, we can apply~\cite{khuller}
and obtain that following 
approximation:
with an augmentation factor of 2, we
can serve all requests with at most a logarithmic
factor more middleboxes.
\begin{theorem}\label{thm:groups} 
The weighted resp.~group network function placement problem can be 
$(2, 1+\ln n)$-approximated via the Generalized Greedy 
Algorithm in polynomial time, i.e., 
the number of deployed middleboxes is at most $(1+\ln n)$ times 
the minimum and the maximum load on each middlebox is at most $2\kappa$.
\end{theorem}

\noindent\textbf{Relationship to Datacenter Scheduling.}
Let us now elaborate more on the interesting relationship of this problem to datacenter scheduling problems. 
In the datacenter scheduling problem~\cite{khuller},
we are given a set of $m$ machines and $n$ jobs.
The processing time of job $j$ on machine $i$ is 
$p_{i,j}$. Each machine $i$ has an activation
cost $a_i$. The scheduling algorithm is given a
total activation budget $A$, and needs to compute
a subset of machines to activate which respect this
budget constraint, and minimizes the \emph{makespan} $T$: 
the length of the execution schedule.

We translate this problem to our problem as follows. 
The machines correspond to the $m$ possible middleboxes
(set $M$ corresponds to all possible locations $U$),
and the jobs correspond to the $n$ requests.
The cost of assigning a group resp.~a weighted 
request $j$ to a middlebox $i$ corresponds to $p_{i,j}$;
since in our model, this cost only depends on the request $j$ but
not on the middlebox $i$, we simply write $p_{j}$ instead of $p_{i,j}$.
The activation cost $a_i$ corresponds to the deployment cost
and is the same for all middleboxes; we set $a_i=1$.

In our algorithm, we greedily aim to select a subset of servers / middleboxes
which maximize the (weighted sum of) covered requests.
We need to cover all the requests, subject to middlebox capacity constraints
$\kappa$, which correspond to the makespan.

\section{Evaluation}\label{sec:eval}

\begin{figure}[b!]
\centering 
\begin{minipage}{0.51\textwidth}
\centering 
\begin{tabular}{c|c|c|c}
Name  & Type & $|V|$  & $|E|$ \\ \hline
Quest 		& Continent  & 20 & 62 \\
GtsHungary 	& Country  & 30 & 62 \\
Geant 		& Continent & 40 & 122 \\
Surfnet 	& Country  & 50 & 136 \\
Forthnet 	& Country  & 62 & 124 \\
Telcove 	& Country  & 71 & 140 \\
Ulaknet 	& Country 		& 82 & 164 
\end{tabular}
\bigskip
\end{minipage}
\begin{minipage}{0.43\textwidth}
\centering 
\begin{tabular}{c|c|c|c|c}
Name  & Type & $|V|$  & $|E|$ & $|D|$\\ \hline
cost266 		 & Continent  & 37 & 55  &  1,332\\
germany50    & Country  & 50 & 88    &  662 \\
nobel-eu    	& Continent  & 28 & 41   &  378 \\
ta2 			& Country & 65 & 108     &  1,869\\
\end{tabular}

\end{minipage}

\caption{Topology Zoo Instances (top) as well as SNDlib instances (bottom) used in the evaluation. $|D|$ denotes the number of defined end-to-end communication requests.}
\label{tab:topologies}
\end{figure}

In order to complement our formal analysis,
we conduct a simulation study comparing the approximation algorithms presented above to the respective optimal algorithms implemented using Integer Programming (IP).
We evaluate our greedy approximation algorithm using the following metrics: (1) the number of installed middleboxes and (2) runtime.
Furthermore, we explore the performance of the 
approximation algorithm for the weighted variant and when incrementally deploying single middleboxes.

All experiments were conducted on servers with two Intel XEON L5420 processors (8 cores overall) equipped with 16 GB RAM.

\subsection{Datasets}
\label{sec:datasets}

We have generated two datasets, one for the unweighted approximation algorithm presented in Section~\ref{sec:approximation} and one for the weighted approximation algorithm presented in Section~\ref{sec:groups}.

For the unweighted variant, we use real-world wide-area topologies 
obtained from the topology zoo collection~\cite{topologyZoo}. 
As shown in Figure~\ref{tab:topologies} (top), the topologies are either 
country-wide or continent-wide ISP, backbone, or -- in case of Geant --  
research networks. The topologies were selected so 
that the number of nodes is equally spread across the range 20 to 82. 
Furthermore, the topologies provide detailed geographical information for nodes, 
such that communication latencies can be estimated based on the
geographical distance. 

For each graph, we generate instances as follows:
Each of the $|V| \cdot (|V|-1) / 2$
potential communication pairs is selected with probability $p$,
set to $0.2$, $0.3$, or $0.4$, respectively. Hence, the number of expected communication pairs to be created overall is $p \cdot |V| \cdot (|V|-1) / 2$.

To ensure comparability
across topologies, all nodes are allowed to host middlebox functionality.
The capacity $\kappa$ is set to $\lceil 2 \cdot (|V|-1) \cdot p\rceil$:
Hence, (in expectancy) around $|V|/4$ many middleboxes are likely to be necessary to serve the communication pairs.
All communication pairs have the same
(latency) stretch. Concretely we set the maximal allowed latency
to $\{1.00, 1.05, 1.10, \dots 2.50\}$ times the latency of the shortest
path between the nodes of the communication pair. For each topology and each probability,
we generated $11$ scenarios uniformly at random, yielding more than 7,000 instances overall.

We employ a similar approach for the evaluation of the 
approximation algorithm for the weighted variant. Concretely, we consider four instances (see Figure~\ref{tab:topologies} (bottom)) of the SNDlib~\cite{Orlowski2010}, which already define communication pairs exhibiting a diverse demand structure. Analogously, we  again consider stretches of $\{1.0, 1.05, \dots, 2.5\}$. To generate multiple instances from a single SNDlib instance, we sample requests by selecting each original communication pair independently at random with a probability of 0.5. We again allow to place middleboxes on all nodes in the network and set the capacity of each of the potential middlebox locations to $4 \cdot D / |V|$, where $D$ denotes the cumulative demand of the chosen requests, such that at least $1/4$ many nodes must be equipped with middleboxes.
We generate 3,100 scenarios in total by considering 31 different stretch parameters and considering 25 random generated instances for the four different topologies.

\subsection{Baseline Algorithms}\label{sec:01}

Besides employing the approximation algorithms presented above, we use Integer Programs to compute optimal solutions. Integer Program (IP)~\ref{alg:baseline} computes optimal solutions for the middlebox deployment problem with unitary demands\footnote{
The existence of a 0-1 integer linear program,  
together with our NP-hardness result, also proves the NP-completeness~\cite{Karp72}.} and works as follows:

For all potential middlebox locations $u\in U$, let $S_u$ be the
set of pairs $P$ that can be routed through $u$ on a path
of stretch at most~$ \rho$,
i.e., $S_u = \{p=(s,t)\in P$ : $d(s,u)+d(u,t)\leq \rho \cdot d(s,t)\}$,
where $d(v,w)$ denotes the length of the shortest path between
nodes $v,w\in V$ in the network.
$S_u$ can be precomputed efficiently.
For all potential middlebox locations $u\in U$, we introduce the
binary variable $x_{u}\in \{0,1\}$.
The variable $x_u$ indicates that $u$ is selected as a
middlebox node in the optimal solution $M$, i.e.
$x_u = 1 \Leftrightarrow u\in M$.
For all $u\in U$ and $p\in P$, 
we introduce the binary
variable $x_{up}\in \{0,1\}$.
The variable $x_{up}$ indicates that the pair $p=(s,t)\in P$
is assigned to the node $u\in U$, s.t.~the path stretch from
$s$ to $t$ through $u$ is at most $\rho$.


\begin{figure}[t]
 {
 \removelatexerror
 
 \def\arraystretch{1.4}
  \LinesNotNumbered
  \begin{IPFormulation}{H}
  \SetAlgorithmName{IP}{}{{}}
\noindent
\begin{tabular}{RFRLQR}
 & \multicolumn{4}{L}{~~~~\textrm{min} ~ \sum_{u\in U} x_u } &  \tagIt{P01_1}\\
\textnormal{subject to} & \sum_{u\in U} x_{up} & = & 1 \qquad & \forall\ p\in P & \tagIt{P01_2}\\ 
& \sum_{p\not\in S_u} x_{up} & = & 0 & \forall\ u\in U & \tagIt{P01_3}\\
& \sum_{p\in P} x_{up} & \leq  & \kappa \cdot x_u &  \forall\ u\in U & \tagIt{P01_4}\\
& x_u, x_{up} & \in & \{0,1\} \qquad ~ & \forall\ u\in U, p\in P ~~~~~~& \tagIt{P01_5}
\end{tabular}
  \caption{Baseline for Minimizing the Number of Middleboxes}
  \label{alg:baseline}
  \end{IPFormulation}
  
  }
\end{figure}


\begin{figure}[t!]
 {
 \removelatexerror
 
 \def\arraystretch{1.4}
  \LinesNotNumbered
  \begin{IPFormulation}{H}
  \SetAlgorithmName{IP}{}{{}}
\noindent
\begin{tabular}{RFRLQR}

& \multicolumn{4}{L}{\textrm{max} ~ \sum_{u\in U} x_{up}} &  \tagIt{P01_1_cp}\\
\textrm{subject to} &  \sum_{u\in U} x_{up} & \, \leq \,  & 1 & \forall\ p\in P & \tagIt{P01_2_cp}\\ 
& \sum_{u\in U} x_u & \, = \, & n & \forall\ p\in P & \tagIt{P01_new_cp}\\
& \sum_{p\not\in S_u} x_{up} &\,  = \, & 0  & \forall\ u\in U & \tagIt{P01_3_cp}\\
& \sum_{p\in P} x_{up} & \, \leq \, & \kappa \cdot x_u  ~~~~~~~~~~ &\forall\ u\in U & \tagIt{P01_4_cp}\\
& x_u, x_{up} &\,  \in \, & \{0,1\}   &\forall\ u\in U, p\in P & \tagIt{P01_5_cp}
\end{tabular}
  \caption{Maximum Assignment for exactly $n$ Middleboxes}
  \label{alg:maximum-assignment-for-n-middleboxes}
  \end{IPFormulation}
  }
\end{figure}

The Objective Function (\ref{P01_1}) requires that a minimum
cardinality middlebox set must be selected.
Constraints (\ref{P01_2}) declare that each pair $p=(s,t)\in P$
is assigned to exactly one node $u\in U$.
Constraints (\ref{P01_3}) state that each pair $p=(s,t)\in P$ can
only be assigned to a node $u\in U$ with $p\in S_u$. By the definition
of $S_u$, the pair $p$ can be routed through $u$ via a path
of stretch at most $\rho$.
Constraints (\ref{P01_4}) describe that the capacity limit $\kappa$ must not
be exceeded at any node, and nodes $u\in U$ which are not selected in the
solution $M$ (where the corresponding variable $x_u$ becomes 0 in the solution)
are not assigned to any pair $p\in P$.

This Integer Program can easily be adapted for the weighted (and or group) variant by adapting Constraint (\ref{P01_4}). Concretely, Constraint~\ref{P01_4} needs to include the additional factor $d_p$, denoting the demand of request $p$:
\[
\sum_{p\in P} d_{p} \cdot x_{up} \leq \kappa \cdot x_u \qquad \forall u\in U~.
\]
As we are also interested in studying the opportunities arising in incremental deployment scenarios, we also introduce IP~\ref{alg:maximum-assignment-for-n-middleboxes} which computes the maximum assignment for any given number of middleboxes.  Concretely, 
given a number $n \in \mathbb{N}$ of middleboxes to activate (see Constraint~\ref{P01_new_cp}), the number of connected communication pairs is maximized (see Constraint~\ref{P01_1_cp}), while ensuring that a communication pair may at most be assigned to exactly one middlebox (see Constraint~\ref{P01_2_cp}).


We have implemented the IPs in Python using Gurobi 6.5.

\subsection{Runtime and Number of Middleboxes\label{sec:runtime}}
%
%

We first study the runtime and the empirical approximation factor of the approximation algorithm~\ref{alg:augment}  on the randomly generated topology zoo instances. To this end, we have implemented the greedy approximation algorithm in Python. Our implementation may use multi-threading for the computation of the middlebox-selection in Algorithm~\ref{alg:augment} (Lines 4-9). The computation of the Integer Program solutions is implemented via Gurobi 6.5 using a single thread.

\begin{figure*}[tbh]
\centering
\includegraphics[width=0.648\columnwidth]{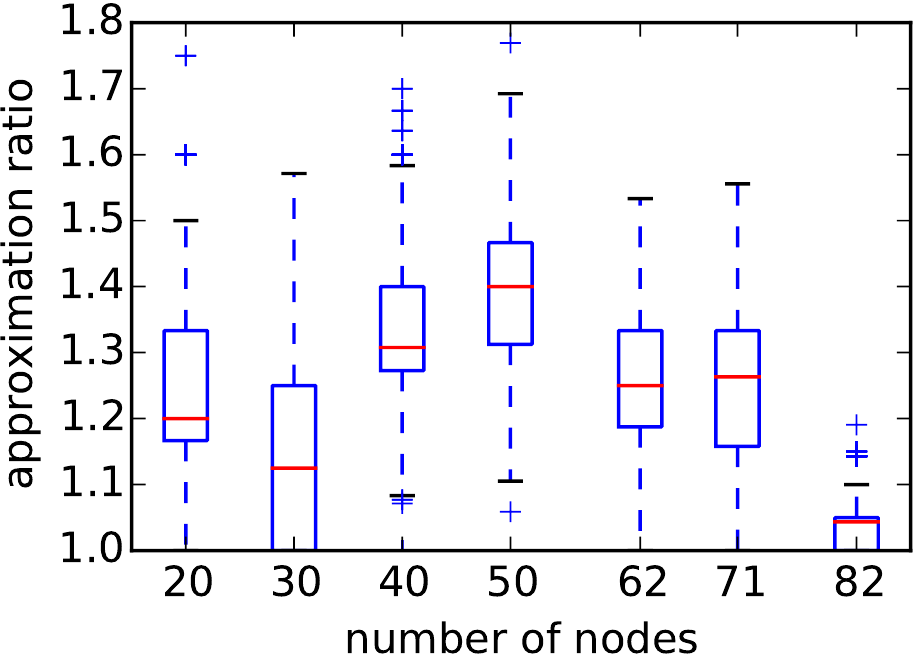}~~~
\includegraphics[width=0.6\columnwidth]{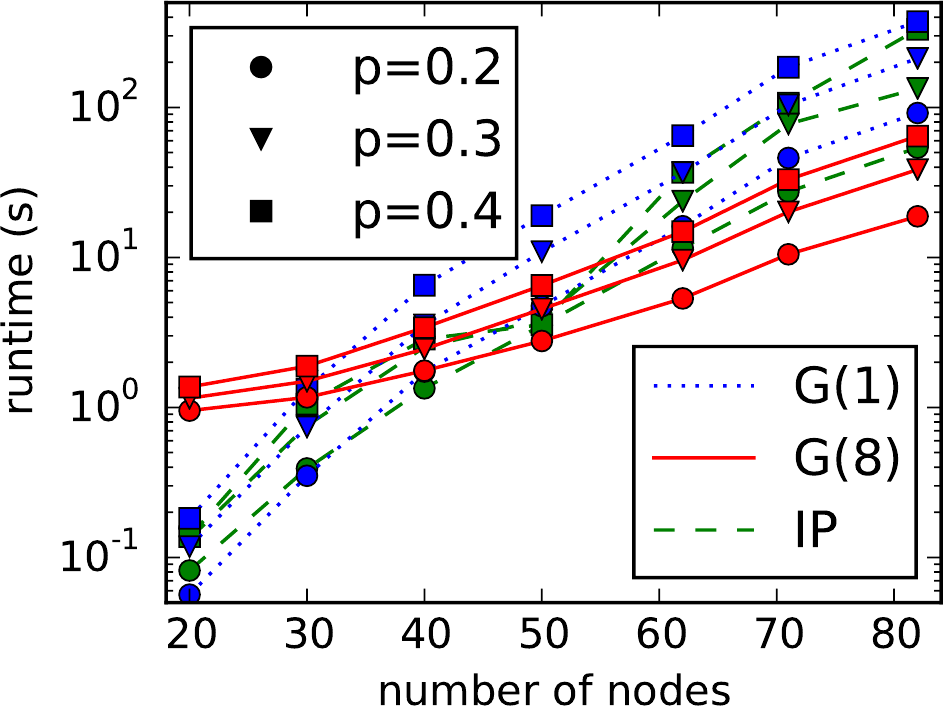}~~~
\includegraphics[width=0.6\columnwidth]{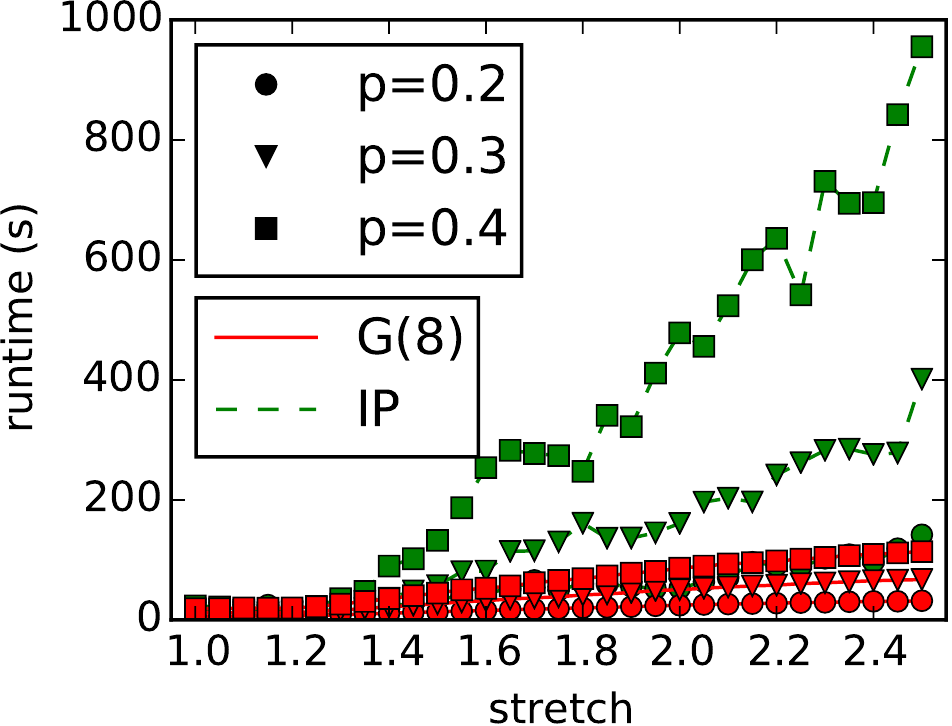}
\caption{\emph{Left:} Empirical approximation ratio. \emph{Middle:} Runtime of the algorithms as a function of the different topologies. \emph{Right:} Runtime as a function of the stretch on the Ulaknet topology.}
\label{fig:runtime} 
\end{figure*}

In Figure~\ref{fig:runtime} (right and bottom), the average runtime
of the greedy (G) approximation algorithm with 1 and 8 threads,
and the IP is depicted. In the middle plot the averaged runtimes are shown.
Here, each data point represents the aggregate of $31 \times 11 = 341$ experiments.
The average runtime of the sequential greedy-algorithm lies below
the one of the IP for 20 and 30 nodes, and the greedy algorithm with 8 threads
clearly outperforms the IP on the topologies
with 62, 71 and 82 nodes. On the largest topology
the computation of the greedy algorithm can be sped up by a factor of around 5, by using 8 threads. Furthermore, the runtime of the IP is on the largest topology one magnitude
higher than the one of the 8-threaded greedy algorithm.
The left plot of Figure~\ref{fig:runtime} depicts the runtime of the
8-threaded greedy variant and the IP on the largest topology as a function of the stretch. Starting at a stretch of 1.3, the runtime of the IP increases dramatically.
This is due to the fact that by increasing the stretch, the number of potential middleboxes, serving a communication pair, increases. In fact, on the largest topology, the IP consisted of up to 90k variables.

Regarding the number of installed middleboxes, the left box plot of Figure~\ref{fig:runtime} shows the approximation ratio, i.e., the number of middleboxes opened by the greedy algorithm divided by the optimal number of middleboxes computed by the IP. The median lies below 
1.5 and the maximum is close to 1.8.

\subsection{Weighted Requests}


Next, we study the performance of the approximation algorithm discussed in Section~\ref{sec:groups}, for weighted problems (cf. Algorithm~\ref{alg:gen-augment}). As discussed in Section~\ref{sec:datasets}, we use SNDlib instances exhibiting a diverse demand structure (see Figure~\ref{fig:group-evaluation-1-2} \emph{left}). Again, we have implemented the generalized approximation algorithm and the corresponding adaption of Integer Program~\ref{alg:baseline} in Python.

As a first result, the middle plot of Figure~\ref{fig:group-evaluation-1-2} shows the averaged relative number of deployed middleboxes of the greedy algorithm with respect to the optimal solution of the IP. The greedy approximation algorithm only seldomly opens -- on average -- more than 20\% middleboxes more than the optimal algorithm. Note that, in some cases, the number of deployed middleboxes lies even beneath the optimal one. This is possible, as the considered algorithm may (cf.~Figure~\ref{fig:group-evaluation-1-2} \emph{right}) violate the middlebox capacity up to a factor of 2. Indeed, the
approximation algorithm violates capacities in less than 25\% of the cases (cf.~Figure~\ref{fig:group-evaluation-1-2}).

\begin{figure*}[tbh]
\centering
\includegraphics[width=0.64\columnwidth]{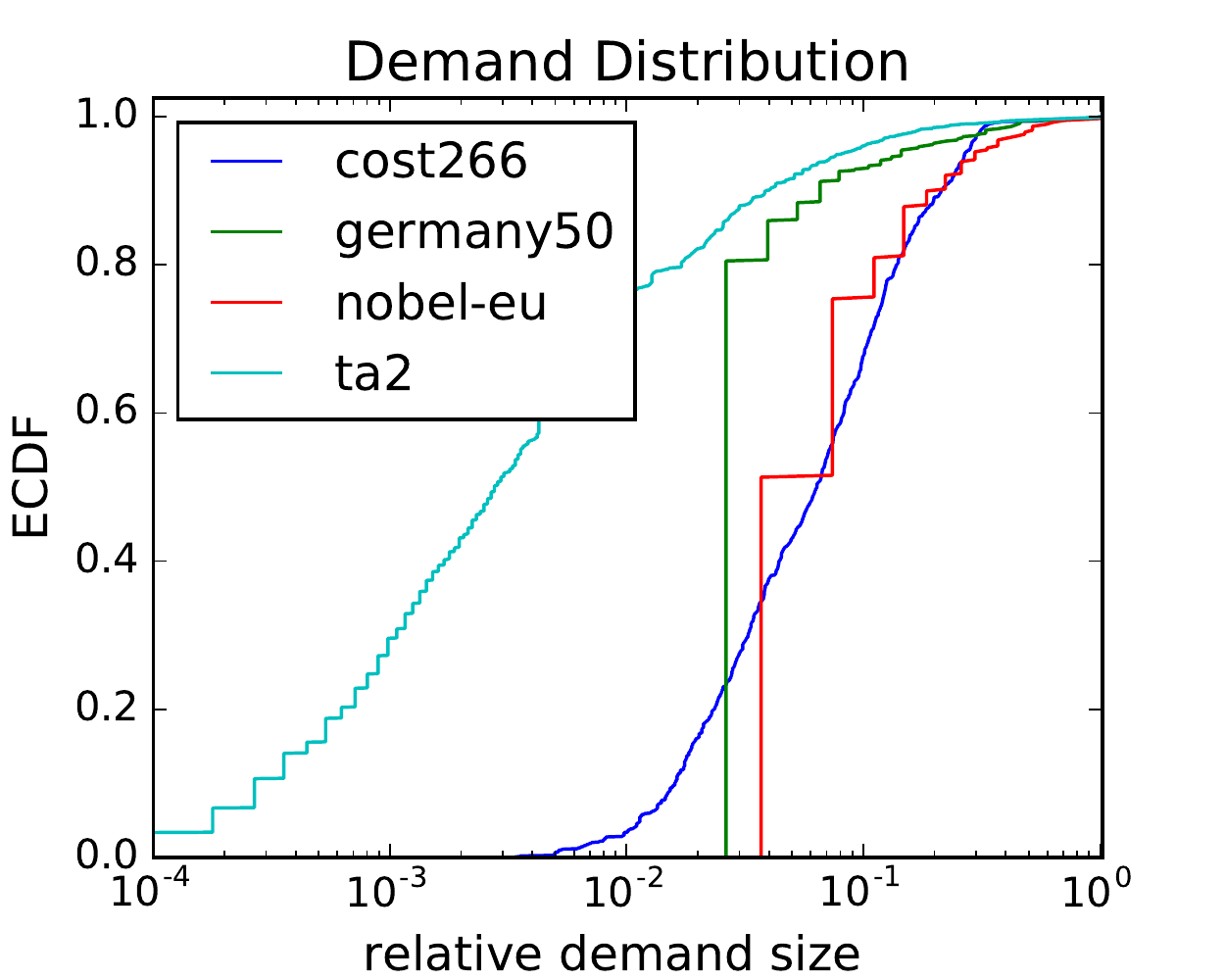}~~~~
\includegraphics[width=0.64\columnwidth]{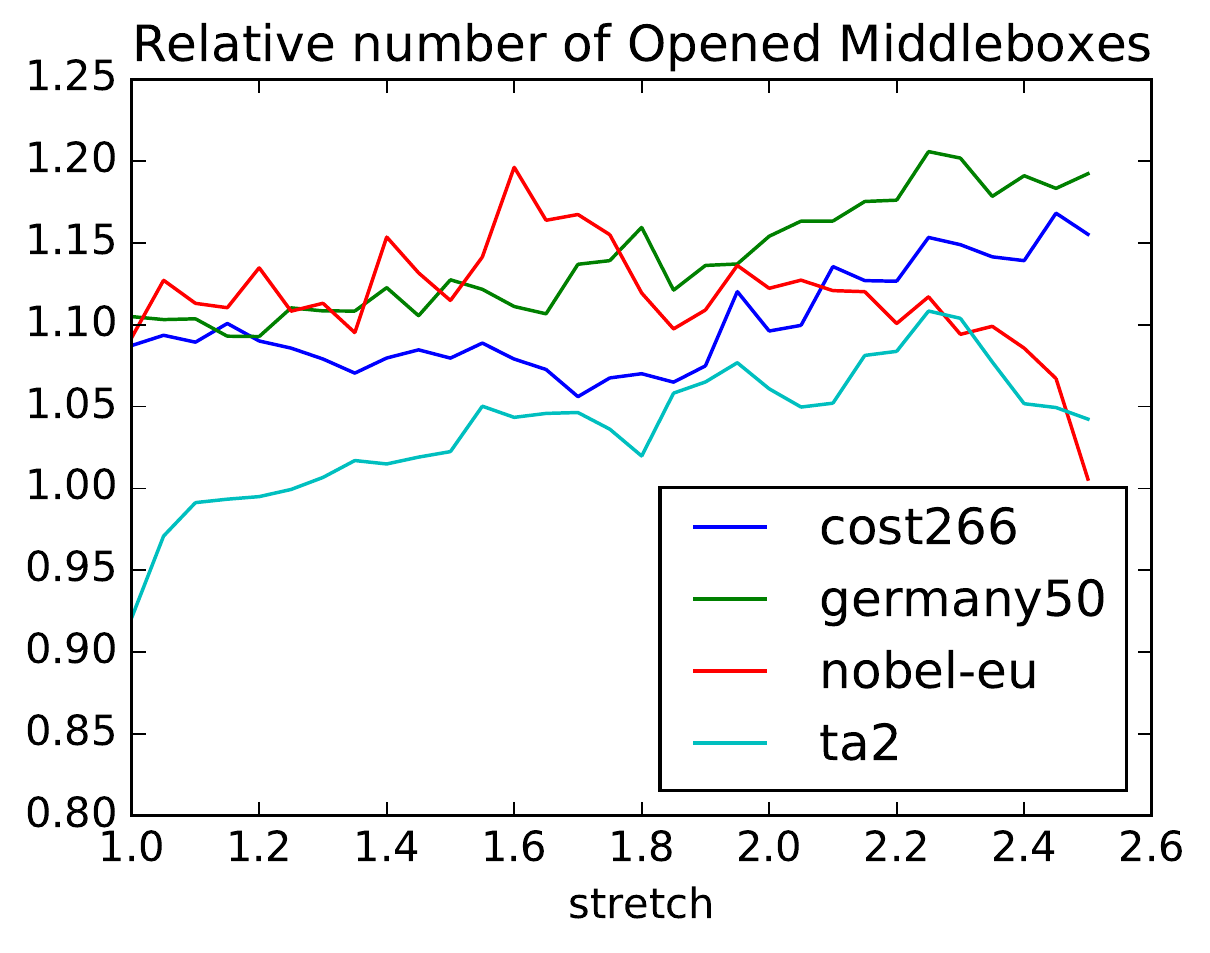}~~~~
\includegraphics[width=.64\columnwidth]{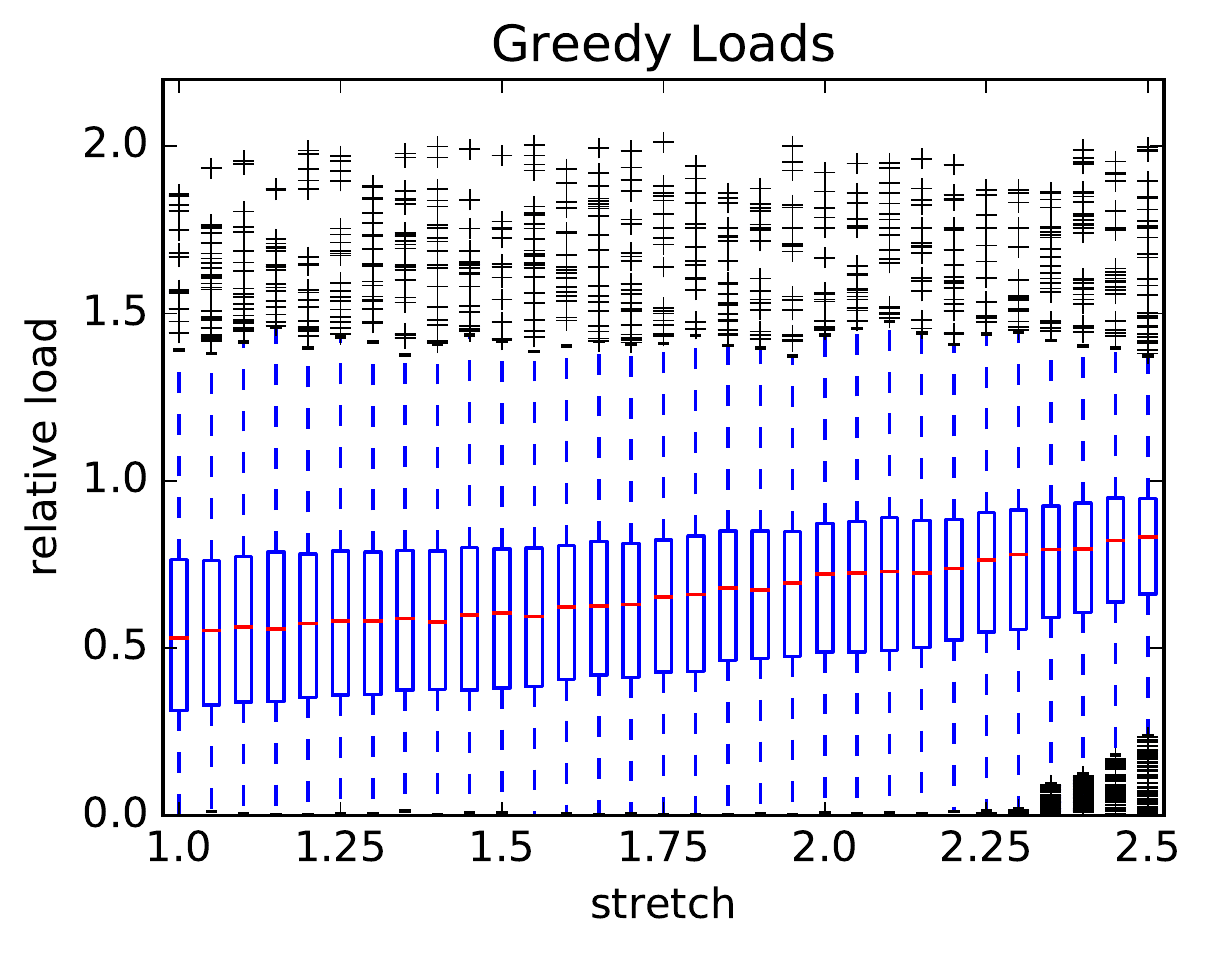}
\caption{\emph{Left:} The demand structure on the SNDlib instances as empirical cumulative distribution function. \emph{Middle:} The relative number of opened middleboxes compared to the solution of the baseline IP. \emph{Right:} The relative load, i.e. load divided by the available capacity, on the middleboxes computed by the $(2,1+\ln n)$-approximation algorithm.}
\label{fig:group-evaluation-1-2}
\end{figure*}

\subsection{Incremental Middlebox Activation\label{sec:incremental-middlebox-activation}}
%
%
Lastly, we study scenarios in which middleboxes are added one after another. Concretely, we assume that all communication pairs are known in advance while the network operator can only incrementally install single middleboxes (e.g. due to the associated cost). The greedy algorithm always places one additional middlebox while not being allowed to change previously selected middlebox locations. We study the number of assigned communication pairs of the greedy algorithm compared to the optimum number of assignments when middlebox locations may be arbitrarily selected. To compute the optimum number of assignments the IP~\ref{alg:maximum-assignment-for-n-middleboxes} is used.

We consider the same experimental setup as described in Section~\ref{sec:runtime} (i.e. topology zoo instances), while constraining the probability to create communication pairs to $p=0.3$. 

We present the results of this set of experiments in Figure~\ref{fig:incremental-mb-after-mb}. The left plot depicts the relative difference of assigned communication pairs when using the greedy algorithm compared to the IP baseline. Concretely, the relative difference is defined as $(\phi_{\mathrm{IP}}-\phi_{\mathrm{G}})/\phi_{\mathrm{IP}}$, where $\phi_{\mathrm{IP}}$ and $\phi_{\mathrm{G}} $ denotes the number of assigned communication pairs of the IP and the greedy algorithm, respectively.
The rows of the left plot averages the results for all scenarios having the same number of minimal middleboxes to serve all communication pairs. The number of scenarios averaged in this way is depicted in the right plot of Figure~\ref{fig:incremental-mb-after-mb}: the number of scenarios for which 10 middleboxes suffice and which are averaged in the row 10 is for example around 520.

Considering any row, we see that the greedy algorithm assigns nearly always as many communication pairs until the number of middleboxes reaches the optimum (minimal) number of middleboxes. After coming close to the optimum number of middleboxes, the relative difference in assignments reaches a maximum valiue of $0.15$ and then diminishes only slightly with each additional greedily placed middlebox. The ability to (re-)place middleboxes in an arbitrary fashion hence only becomes important when sufficiently many middleboxes were already places to serve almost all communication pairs.

\begin{figure}[h!]
\centering
\includegraphics[width=.50\columnwidth]{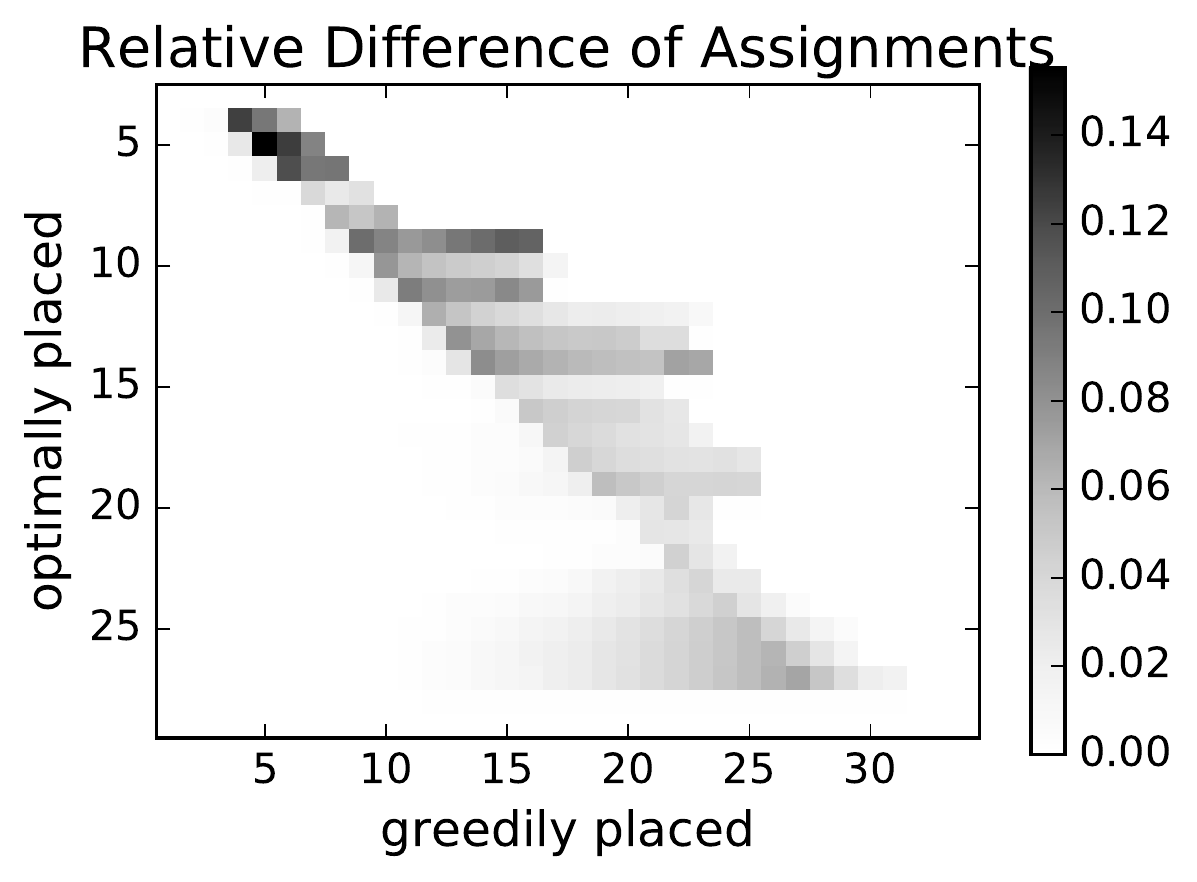}
\includegraphics[width=0.47\columnwidth]{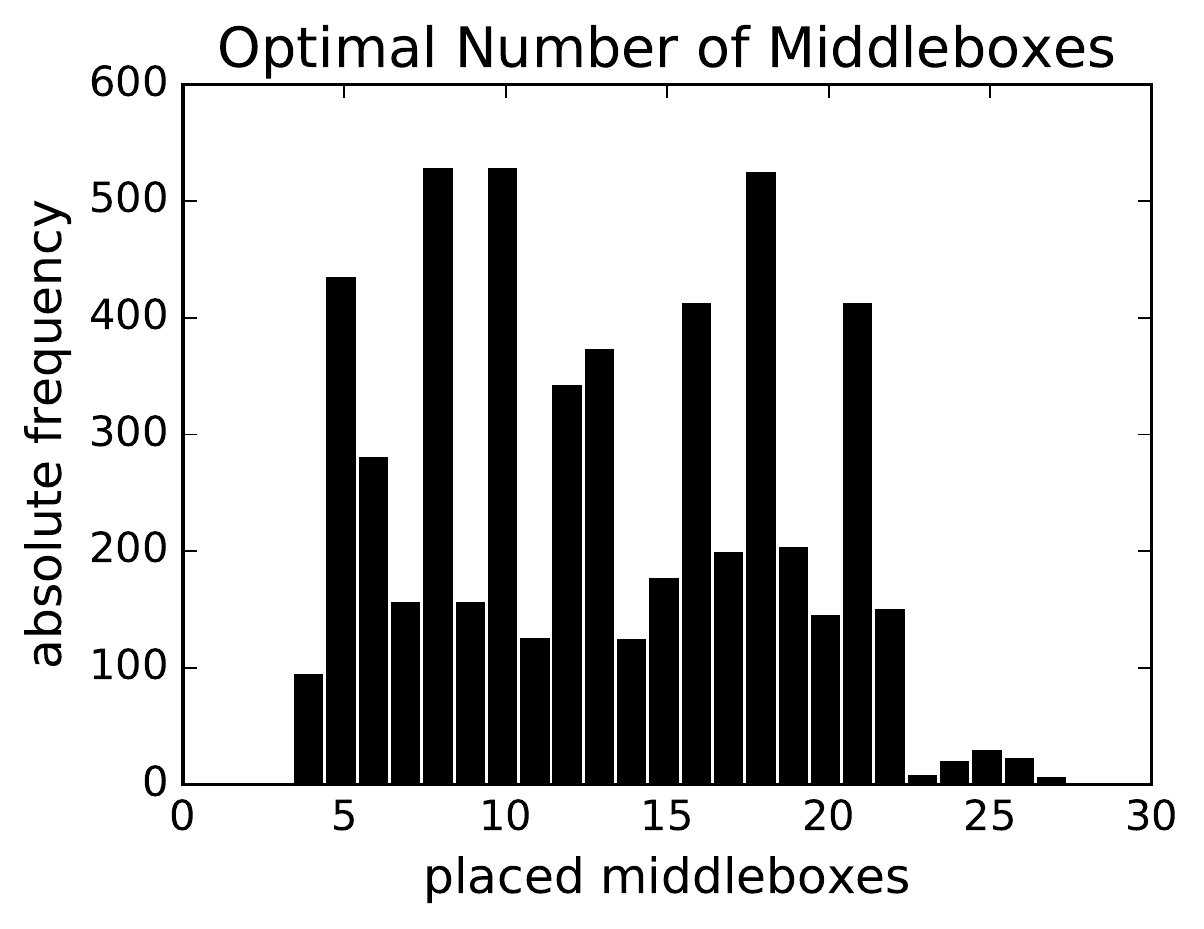}
\caption{\emph{Left:} The averaged relative 
difference in the number of served communications pairs depending on the optimal number of middleboxes (y-axis). \emph{Right:} Minimum number of middleboxes computed by the IP for the scenarios averages in the left plot. } 
\label{fig:incremental-mb-after-mb}
\end{figure}

\section{Summary and Conclusion}\label{sec:summary}


This paper initiated the study of the network function placement
problem which is motivated
by the increasing flexibilities of modern 
virtualized networked systems. Our main contribution is a 
combinatorial proof of the submodularity of this problem and 
an incremental log-approximation 
network function placement algorithm.
We also initate the study of a randomized rounding approach for a 
weighted group-version of the problem.
Our simulation results show that 
this approach computes a nearly optimal placement for real world network instances.
  
We understand our work as a first step,
and believe that our paper opens several interesting directions
for future research. 
In particular, it will be interesting to know whether good approximations exist for the incremental deployment of entire group requests.  



{
\bibliographystyle{abbrv}
\bibliography{typeinst}
}

\end{document}